\documentclass[acmsmall,nonacm]{acmart}
\usepackage{graphicx}
\usepackage[T1]{fontenc} 
\usepackage[dvipsnames]{xcolor}
\usepackage{framed}
\usepackage{comment}
\usepackage{amsthm}
\usepackage{amsmath,amsfonts}
\usepackage{mathtools}
\usepackage{breakurl}           
\usepackage{url}                
\usepackage{enumitem}

\usepackage{tikz}
\usetikzlibrary{arrows.meta,positioning,fit,calc}

\usepackage[capitalise]{cleveref}
\usepackage{algorithm} 
\usepackage{algpseudocode}
\usepackage{float}
\usepackage{caption}
\usepackage{thm-restate}
\usepackage{xspace}
\usepackage{xurl}
\usepackage[T1]{fontenc}
\usepackage{subcaption}
\newcommand{\block}{\mathcal{B}} 
\newcommand{\reward}[1]{\mathcal{R}(#1)} 
\newcommand{\tran}{\tau} 
\newcommand{\tx}{\mathsf{tx}} 

\newcommand{\sstate}{\mathcal{S}} 
\newcommand{\st}{\mathsf{st}} 

\newcommand{\adv}{\mathcal{A}} 
\newcommand{\gas}{g} 
\newcommand{\maxgas}{\gas_{\mathsf{max}}} 
\newcommand{\gs}{\gas^*} 

\newtheorem{definition}{Definition}

\newcommand{\ray}[1]{{ [\color{orange} Ray: #1]}}

\newcommand{\ahad}[1]{{ [{\color{purple}{\sf Ahad:} #1}]}}

\newcounter{compactcorollary}
\renewcommand{\thecompactcorollary}{\arabic{compactcorollary}}

\crefname{compactcorollary}{corollary}{corollaries}
\Crefname{compactcorollary}{Corollary}{Corollaries}

\newcommand{\est}{\mathsf{Est}}

\newcommand{\cost}{\mathsf{\mathbf{Cost}}} 

\newcommand{\TaskSet}{\mathbb{T}}
\newcommand{\crRes}{\mathsf{CR}_\mathsf{Res}}
\newcommand{\crRew}{\mathsf{CR}_\mathsf{Rew}}

\def \ifempty#1{\def\temp{#1} \ifx\temp\empty }

\newcommand{\E}{\mathrm{E}}

\newcommand{\var}[1]{\textit{#1}}
\newcommand{\op}[1]{\textsl{#1}}

\newcommand{\CA}{\ensuremath{\mathcal{A}}\xspace}

\newcommand{\CR}{\ensuremath{\mathcal{R}}\xspace}

\title{The Tragedy of Chain Commons}
\subtitle{Decoupled Blockchains are Unfair}

\author{ Ignacio Amores-Sesar}
\affiliation{%
  \institution{Aarhus University}
  \country{Denmark}
}
\email{amores-sesar@cs.au.dk}

\author{Mirza Ahad Baig}
\affiliation{%
  \institution{Institute of Science and Technology Austria}
  \country{Austria}
}
\email{mbaig@ista.ac.at}

\author{Seth Gilbert}
\affiliation{%
  \institution{National University of Singapore}
  \country{Singapore}
}
\email{seth.gilbert@comp.nus.edu.sg}

\author{Ray Neiheiser}
\affiliation{%
  \institution{Institute of Science and Technology Austria}
  \country{Austria}
}
\email{ray.neiheiser@ist.ac.at}

\author{Michelle X. Yeo}
\affiliation{%
  \institution{Aarhus University and Nanyang Technological
  University}
  \country{Denmark}
}
\email{michellexyeo@gmail.com}

\begin{abstract}
Byzantine Fault Tolerant (BFT) consensus forms the foundation of many modern blockchains striving for both high throughput and low latency. 
A growing bottleneck is transaction execution and validation on the critical path of consensus, which has led to modular decoupled designs that separate ordering from execution: consensus orders only metadata, while transactions are executed and validated concurrently. 
While this approach improves performance, it can leave invalid transactions in the ledger, increasing storage costs and enabling new forms of strategic behavior. We present the first systematic study of this setting, providing a formal framework to reason about the interaction between consensus and execution. Using this framework, we show that the decoupled design enables a previously unidentified attack which we term gaslighting. We prove a fundamental trade-off between resilience to this attack and resource capacity utilization, where both are impossible to achieve deterministically in the decoupled model. To address this trade-off, we discuss an intermediate model for leader-based protocols that is robust to gaslighting attacks while achieving high throughput and low latency.
\end{abstract}



\begin{document}

\maketitle

\section{Introduction}

The pursuit of high-performance blockchains has led to a series of novel protocols in recent years~\cite{narwhal,kauri,anthemius,blockstm} with the goal to achieve comparable performance to traditional financial systems, while achieving other desirable properties like decentralization and privacy.
A common feature of many of these new approaches is the \emph{decoupling} of consensus, block creation, and execution, often referred to as lazy or modular blockchains~\cite{lazyledger}. 
This architectural change has gained popularity as it substantially increases the execution capacity of the system, which has recently been identified as a major performance bottleneck~\cite{chiron,anthemius}.

In the coupled setting, consensus, execution, and block creation are strictly sequential. 
A block $B$ is proposed, agreed upon via consensus process $C$, and, upon successful agreement, executed. 
The next block cannot be proposed until execution process $E$ completes, placing execution on the critical path of consensus and causing longer execution times to directly increase latency.

In contrast, in the decoupled setting, consensus, execution, and block creation can proceed in parallel. Execution is no longer constrained to the interval between two consensus rounds, but can advance concurrently to consensus and block creation. This substantially increases the time available for execution, enabling higher throughput without affecting consensus latency. We highlight the differences between the two approaches in Figure~\ref{fig:coupleddecoupled}, where the decoupled model triples the available execution capacity. We also observe that in leaderless blockchains, such as DAG-based protocols~\cite{DBLP:conf/podc/KeidarKNS21,narwhal,mirbft,DBLP:conf/wdag/Amores-SesarGHO25,DBLP:conf/ndss/BabelCDKKKST25}, decoupling consensus from execution is not merely an architectural choice but a necessity, as it is essential to allow proposers to propose blocks in parallel. 


\begin{figure}[t!]
    \centering    \includegraphics[width=0.5\linewidth]{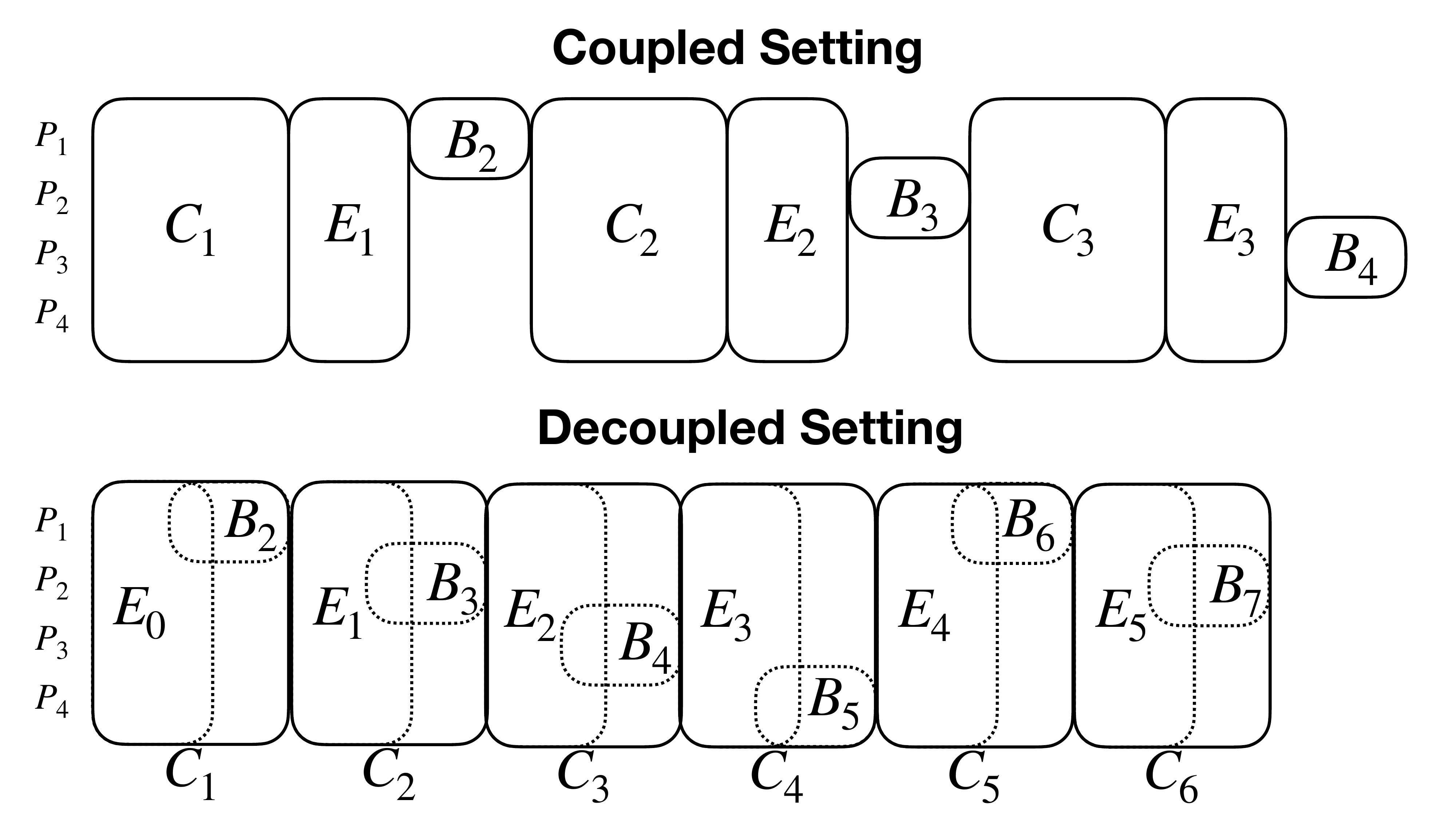}
    \caption{By moving execution and block construction from the critical path of consensus, the decoupled mode can produce blocks  faster in the same time frame  while also increasing the execution capacity. However, the parties do not know the exact state of the ledger at the time of creating a block}
    \label{fig:coupleddecoupled}
\end{figure}

However, this increase in execution capacity comes with an important trade-off:  
since execution is decoupled from consensus, the current block proposer does not have access to the full up-to-date state of the blockchain at the time of block creation, as this state can only be known after the execution of the previous block is completed.
For instance, in~\Cref{fig:coupleddecoupled}, the proposer of block $B_2$ in the coupled setting knows the full state of the blockchain as $B_2$ is proposed after $E_1$ (i.e., the execution of $B_1$) completed.
However, in the decoupled setting, block $B_2$ is proposed before $E_1$ is completed.
As a result, decoupled blockchains have a \textit{dirty ledger}~\cite{lazyledger}, where transactions may fail due to insufficient funds, or may even be included multiple times across different blocks.

More fundamentally, the lack of knowledge of the full state in the decoupled setting introduces unique challenges for blockchain mechanism design. 
A central component of blockchain mechanism design is that users pay a fee for their transactions to be included in a block by the proposer. This fee, commonly referred to as \emph{gas}, compensates for the resources consumed during transaction execution, including storage and computation. 
Furthermore, blocks are subject to a max gas capacity that bounds the total resource usage of the included transactions. Ideally, the block capacity should saturate the available system resources (e.g., CPU), filling each block up to this capacity with transactions. 
However, in the decoupled setting, the proposer cannot accurately determine at block creation time how much resources a transaction will actually consume during execution.

Based on this observation, we identify a new attack that is unique to the decoupled setting. 
The attack, which we term \emph{gaslighting}, is launched by malicious users who post transactions that appear to be costly but end up being cheap, thereby ``gaslighting'' the victim block proposer.
This attack exploits the fact that in decoupled blockchains the proposer lacks knowledge of the full state at block creation\footnote{We stress that we are primarily interested in architectural trade-offs and mechanism design, and not in computational optimization problems, such as the knapsack problem arising when constructing an optimal block from a set of transactions.}, and violates a fundamental desideratum of blockchains: \emph{fairness}, i.e., that rewards are distributed in a fair manner to block proposers, which consequently has a negative long-term impact on the security of the blockchain. This leads to our central research question: 

\smallskip{}
\textit{Are these trade-offs inherent to the decoupled model, or can they be circumvented through careful mechanism design?}
\smallskip{}

We answer this question in the negative.
In fact, in light of these trade-offs and the impossibility results we establish, we do not recommend that leader-based blockchains adopt the decoupled model. Instead, we advocate for alternative coupling models that preserve existing mechanism design guarantees, such as the approach introduced in this paper.


To justify our stance, we first introduce a novel framework to reason about the interplay of consensus and execution and show that there are fundamental tradeoffs inherent to the decoupled setting. We establish our first lower bound in~\Cref{lem:coupledCRRes}: in the decoupled setting, an adversary can launch gaslighting attacks that result in a drastic underutilization of system resources.
~
To make it worse, we show that it is impossible to reliably hold users accountable for this attack and that it can be executed at essentially zero cost.

We then prove a fundamental impossibility in~\Cref{lem:noRewardFairness} for decoupled blockchains: it is impossible to achieve reward fairness among block proposers. In particular, it is impossible to distinguish between (1) honest block construction in the presence of adversarial transactions in the mempool and (2) malicious block construction given an honest mempool. Consequently, reward fairness is impossible to achieve in this setting.

To reap the benefits of both the decoupled and coupled model, we propose a \emph{partially coupled} model that combines the advantages of the two approaches. The key observation underlying our design is that, at block creation time, access to the full blockchain state is unnecessary. Instead, it suffices to know only the subset of state that are modified by transactions selected for inclusion in a given block, as transactions usually modify only a small part of the state. 

Consequently, as long as there are sufficiently many transactions that do not conflict with transactions currently pending in consensus or execution, consensus, execution, and block creation can proceed in parallel. Coupling is required only at the level of individual transactions whose state dependencies overlap.
As a result, our partial coupled model maximizes execution throughput while avoiding the aforementioned attack inherent to the fully decoupled design.

Finally, to demonstrate the practical feasibility of gaslighting attacks, we present a case study showing that overestimation of gas costs is extremely common in decoupled blockchains, and that the economic impact of such attacks could exceed \$300 million.


\smallskip{\bf Our contribution.} We summarize our contribution as follows:

\begin{itemize}
    \item We provide the first formal analysis of the interaction between execution and consensus in decoupled blockchains, and develop a formal framework to reason about their performance and security properties.

   \item Using this framework, we identify a previously overlooked attack vector, which we term \textit{gaslighting}, and prove that it is inherent to the decoupled model~(\Cref{lem:decoupledCRRes,lem:noRewardFairness}).

    \item Furthermore, we show in~\Cref{lem:lowCRRewActualCost,lem:noCostEnforceability}, that users can execute these attacks at essentially zero cost.

    \item We analyze the implications of our theoretical results on practical blockchain architecture design, and introduce a novel \emph{partially coupled} model that is \emph{robust against these attacks} while achieving comparable execution throughput to fully decoupled systems.

    \item Finally, we present a \emph{case study} using real-world data from Sui and Ethereum, demonstrating the practical feasibility of the identified attack, and show the economic damage that could be inflicted on the respective ecosystem. 
    
    
\end{itemize}

\subsection{Related Work} 

\smallskip{\bf Resource allocation.} 
Resource allocation problems have been widely studied in the fields of economics~\cite{mankiw}, mechanism design~\cite{AuYongCSV04,drf,ChowduryLGS16}, operations research~\cite{JohariT04}, algorithms~\cite{rap,CohenKMZ19,DongGS25} and distributed computing~\cite{ReifS82resource,StyerP88}.
Our work extends this rich line of studies and analysis to distributed and decentralized resource allocation, particularly to the blockchain execution setting, where the underlying resource is gas.
A common theme running through these classic resource allocation problems is that the amount of resource needed by each agent to complete their task is known during optimization time, and the goal of the resource allocation problem is to maximize desired properties like social welfare or fairness. 
Our work is most related to recent work that look at the resource allocation problem in the asymmetric~\cite{HarrisT81uncertain} or incomplete information~\cite{MARAR200950} setting, where some costs are unknown at optimization time. 
Unlike these works that use parametric distributions to model uncertainty, we assume completely unknown information  in the blockchain execution setting, and show fundamental impossibilities in this setting.
Additionally, our work relates to prior studies on inefficiencies arising from decentralized system design, often discussed through the lens of the price of anarchy~\cite{JohariT04,selfish}. 
However, in contrast to this line of work, we focus on the concrete interplay between consensus and execution and the fundamental properties resulting from different architectural design choices.

\smallskip{\bf Decoupled blockchains.}
Many recent consensus algorithms rely on the decoupled blockchain model, where ordering and execution of transactions proceed in parallel. 
This approach is taken further by DAG-based and other leaderless protocols where blocks are proposed in parallel without knowledge of the state of concurrent proposals~\cite{hyperledger,DBLP:conf/podc/KeidarKNS21,narwhal,mirbft,DBLP:conf/wdag/Amores-SesarGHO25,DBLP:conf/ndss/BabelCDKKKST25}.
Other recent models have proposed some degree of execution in tandem to consensus, based on models of speculative~\cite{hellingsGRCSS25poe}, optimistic~\cite{sei}, and delayed execution~\cite{eipdelayed}.
Fundamental to decoupling consensus and execution is the notion of the consensus number of an object~\cite{consensusnumber,herlihy91}, where objects with a consensus number of 1 can be trivially decoupled. 
Our work explores the limits of existing mechanisms under the decoupled model, and describes and analyzes new types of attacks that arise from these limitations. 
In addition, we formally introduce a novel model of partial coupling which lies at the intersection of the coupled and decoupled designs and provide an analysis and comparison (to decoupled models) of our partially coupled model with respect to performance and fulfilling desired properties.

\smallskip{\bf Dirty/lazy ledgers.}
The dirty (or lazy) ledger model has motivated several works analyzing its implications, especially for light clients. 
Since the consensus output alone is no longer sufficient to determine validity, new techniques are required to ensure safety for light clients in this setting~\cite{lazyledger,10.1007/978-3-031-78679-2_1,StefoXK23}. 
Our work adopts the aforementioned security assumptions for light clients, as our proposed partially coupled model produces a dirty ledger.

\smallskip{\bf Gas pricing mechanisms.} 
In the blockchain setting, gas pricing and its fundamental role in mitigating performance attacks have been studied extensively from different perspectives, including parallel execution~\cite{liobaparallel}, multidimensional fee markets~\cite{diamandis_et_al:LIPIcs.AFT.2023.4,angeris2024multidimensionalblockchainfeesessentially,eippricing,eipmulti,KiayiasKPZ25,LaveeNPR25}, theoretical analysis of fee mechanisms~\cite{Paramitha_Tarigan_2025,10.3389/fbloc.2024.1462666,leaderlessrewards}, and impact on off-chain bribing\cite{ganesh2025characterizingoffchaininfluenceproof}. 
These works all study and propose gas pricing mechanisms in the coupled setting. 
In contrast, in our work we show fundamental trade-offs that are independent of the gas pricing mechanism design.





\section{Model and Preliminaries}\label{sec:prelims}
\smallskip{\bf Blockchain abstraction.}\label{sec:model}\label{sec:blockchain}
A blockchain protocol is a round based protocol run by a set of $N$ validators $p_1,\ldots,p_N$ and consists of three components: (1) \emph{block creation} $B$, which selects transactions from the mempool, (2) \emph{consensus} $C$, which creates blocks, an ordered set of transactions, and (3) \emph{execution} $E$, which deterministically applies transactions of a block to update the state. For the sake of simplicity, we assume that execution outputs the complete state of the ledger. These components iterate over rounds. We consider protocols $\Pi$ that safely implement atomic broadcast~\cite{bible}, ensuring agreement on a sequence of blocks. We treat atomic broadcast as a black-box abstraction and study two integration models based on the interaction between block creation, consensus, and execution.

First, some basics. A transaction $\tx$ represents a state transition the clients want to do. In our model, each transaction $\tx$ comes with an estimate of maximum resource it is going to consume $\est(\tx)$ and a price $p(\tx)$. We explain the resource measurement and price later in this section. Thus we write a transaction as a tuple $\tx = (\mathsf{content}(\tx),\est(\tx),p(\tx))$.

In a \emph{coupled blockchain}, block creation, consensus, and execution are strictly sequential and lie on a single critical path. Transactions are executed during consensus, and block creation begins only after the previous block is finalized and executed. In mathematical terms, this means that in round $i$, consensus $C_i$ takes as input the output of consensus $C_{i-1}$ and execution $E_{i-1}$ from the previous round and the block from the current round $B_i$ (see~\Cref{fig:coupled}).

\begin{figure}[htbp]
    \begin{minipage}{0.45\linewidth}
        \begin{align*}
            B_i(st_i, M_i) &\to \block_i\\
            C_i(st_i, \block_i) &\to \\
            M_{i+1} &\to (M_i\setminus \block_i) \cup \mathsf{NewTxs}_i\\
            E_i(st_i, \block_i) &\to st_{i+1} 
        \end{align*}
        \caption{Coupled Setting}\label{fig:coupled}
    \end{minipage}
    \hfill
    \begin{minipage}{0.45\linewidth}
        \begin{align*}
            B_i(M_i) &\to \block_i\\
            C_i(\block_i) &\to \\
            M_{i+1} &\to (M_{i}\setminus \block_i) \cup \mathsf{NewTxs}_i\\
            E_i(st_i, \block_i) &\to st_{i+1}
        \end{align*}
        \caption{Decoupled Setting}\label{fig:decoupled}
    \end{minipage}
\end{figure}

In a \emph{decoupled blockchain}, block creation and consensus proceed independently and asynchronously of execution. That is, block creation and consensus do not depend on the most up to date execution output. Execution is asynchronous and off the critical path. Mathematically, in round $i$, consensus $C_i$ takes as input the output of consensus $C_{i-1}$  from the previous round and the block from the current round $B_i$ (see~\Cref{fig:decoupled}). In contrast to the coupled model, consensus is oblivious from the result of execution. Decoupling is commonly used to improve throughput, but its benefits are not well quantified. We model the system using execution time $\delta^{e}$, consensus time $\delta^{c}$, and block creation time $\delta^{b}$. In the coupled model, the slot time is $\delta=\delta^{e}+\delta^{c}+\delta^{b}$. In the decoupled model, execution can utilize the full slot, yielding a speed-up of $\beta=\delta/\delta^{e}$. For example, with $\delta^{c}=600$ ms, $\delta^{e}=200$ ms, and $\delta^{b}=200$ ms, we obtain $\delta=1000$ ms and $\beta=5$.


\smallskip{\bf Execution resources.}
Execution resources represent the costs incurred when processing a transaction, including CPU time, storage, and bandwidth. 
Formally, for transaction space $\tran$ and state space $\sstate$, we define a gas function $\gas : \tran \times \sstate \rightarrow [0,1]$ that maps a transaction $\tx$ and state $\st$ to its cost; we write $\gas(\tx)$ when $\st$ is implicit. This function aims to capture all the costs associated with the transaction life-cycle such as computation, storage, bandwidth, etc. To a reader familiar with blockchains, the output of a gas function is usually in $\mathbb{N}$ but every transaction has a maximum gas limit. To simplify our model and analysis, we can normalize the maximum gas limit per transaction and consider the domain of $[0,1]$ instead. Correspondingly we also require that $\est(\tx) \in [0,1]$. 

\smallskip{\bf Mempool and block creation.} A mempool is a set of pending transactions. Clients submit their transactions to the mempool. Validators propose blocks by picking transactions from the mempool during the block creation process. We denote by $M_i$ the state of mempool at the start of each block creation round. The new transactions submitted by the clients after the start of the previous block creation round ($B_{i-1}$) are denoted by the set $\mathsf{NewTxs}_i$. 

A validator is selected uniformly at random to be the next scheduler who creates the block via Block Creation $B$. Block creation deterministically outputs a set of transactions, called a block $\block_i \subset M_i$, such that $|\block| \leq N$. Additionally, there is a constraint on the total gas for all the transactions in the block. In the coupled model $\sum_{\tx \in \block_i} \gas(\tx,\st_{i-1}) \leq \maxgas$, where $1 \leq \maxgas \leq N$. While in the decoupled model since we do not have access to the state, we require that the $\sum_{\tx \in \block_i} \est(\tx) \leq \maxgas$. Lastly, there is a non-triviality condition that $|\block_i| \geq \min(\lfloor\maxgas\rfloor,|M_i|)$

\smallskip{\bf Transaction and execution model.}
We use the execution model inspired by Ethereum, Solana, Sui, and Aptos~\cite{GrechKJBSS20madmax,aptos,sui,solana}. 
For a transaction $\tx$, let $\gas(\tx, \st)$ denote the actual  gas cost of executing $\tx$ when the state is $\st$. 
If $\gas(\tx,\st) \leq \est(\tx)$, then $\tx$ would be executed.
If $\est(\tx) < \gas(\tx,\st)$, the transaction will be partially executed until the amount $\est(\tx)$ is used up, and then rolled back, i.e., there is no state change from transaction $\tx$. 
In general, blockchains support Turing-complete expressivity in the transactions; we require only that the transactions support if-conditions. 
Thus, it is possible for the same transaction $\tx$ to have $\gas(\tx, \st) \neq \gas(\tx,\st')$ depending on the state of execution.
We also assume that $\st$, the state of the execution, is large and a single transaction can only change a bounded number of bits of $\st$. Formally, consider two blocks $\block$ and $\block'$ such that $\block \triangle \block' = \{\tx\}$, then for any state $\st$, $\delta_{Hamming}((E(\block, \st), E(\block', \st)) << |\st|$.

We abuse the notation to say gas of a block as $\gas(\block, \st) = \sum_{\tx \in \block} \gas(\tx, \st)$.

\smallskip{\bf Client cost model.} To cover resource usage and incentivize proposers, clients are charged a fee per transaction. For a transaction $\tx$ in state $\st$, let $\cost(\tx,\st)$ denote its cost. Currently, all major blockchains charge $\cost(\tx, \st) = c_{base}+p(\tx)\cdot \gas(\tx, \st)$, where $c_{base}$ is a small amount. Think of $c_{base}$ as a small fee to do basic processing of the transaction. We can assume $c_{base} = 0$ as it does not qualitatively hurt our results.

Looking ahead, in our attacks in~\cref{sec:attacks}, we use the current cost model. Later, when we suggest mitigations in~\cref{sec:discussion}, we will consider alternative cost models as well. 

We abuse the notation to say cost of a block is $\cost(\block, \st) = \sum_{\tx \in \block} \cost(\tx, \st)$. 

\smallskip{\bf Enforceability of cost.} A crucial ingredient in some of our attacks is the ability of the adversary to avoid having to pay for the costs incurred in the cost model. This is particular to the decoupled model, and it does so by transferring all the money out of its account while the attack is in progress. We describe this in greater detail in the attack section~\cref{sec:attacks}.    



\smallskip{\bf Adversarial model}
The adversary denoted by, $\adv$, is allowed to control all the transactions in the mempool. 
It can submit an unlimited number of transactions to the mempool at any round. 
Its transactions can have any resource estimate, $\est(\tx) \in [0,1]$ or any $p(\tx) \in \mathbb{R}_{+}$ as long as there are enough transactions with high resource usage. 
The notion of "enough" is defined formally for a congested setting in~\cref{def:congested}.

\subsection{Definitions}\label{sec:desiderata}

Block execution capacity is a common resource with high demand. A highly desirable property is to fully utilize the block capacity in each round.

\begin{definition}[Block Capacity Utilization]
\label{def:bCapUtil}
    For a block $\block_i$ with the execution state $\st_i$, we define block capacity utilization ratio as $\crRes = \frac{\sum_{\tx \in \block_i} \gas(\tx, \st_i)}{\maxgas}$.
\end{definition}

\begin{definition}[Maximum Possible Block Capacity Utilization]
\label{def:bCapUtilMax}
    For mempool in round $M_i$ with the execution state $\st_i$, we define the maximum possible block capacity utilization as $\crRes^* = \frac{\gas(\block_i, \st_i)}{\gas(\block^*_i)}$, where $\block^*_i$ is the block with maximum gas if $\st_i$ is known at the time of block creation.  
\end{definition}

\begin{definition}[Maximum Possible Block Cost Ratio]
\label{def:bCostMax}
    For mempool in round $M_i$ with the execution state $\st_i$, we define the maximum possible block cost ratio as $\crRew^* = \frac{\cost(\block_i, \st_i)}{\cost(\block^*_i, \st_i)}$, where $\block^*_i$ is the block with maximum possible cost if $\st_i$ is known at the time of block creation.
\end{definition}


The block capacity utilization and block cost are affected by the state of the mempool. We want to define two different conditions of the mempool state. 

\begin{definition}[Congested Mempool]
\label{def:congested}
    The mempool for round $i$, $M_i$, is called congested if $|M_i| \geq N$ and there exists $S \subset M_i$ such that $|S|\leq N$ and $\gas(S) = \maxgas$, for the coupled setting. Alternatively, for the decoupled setting a mempool is called congested if $|M_i| \geq N$ and there exists $S, S' \subset M_i$ such that $|S|, |S'| \leq N$, $\est(S) = \maxgas$ and $\gas(S', \st_{i}) = \maxgas$.
\end{definition}

\begin{definition}[Reward]
\label{def:reward}
    The reward for a validator for creating a block $\block$ is $\reward{\block} = \sum_{\tx \in \block} \cost(\tx,\st)$.
\end{definition}

An important desideratum of blockchain protocols is \emph{reward distribution fairness}, which requires that long-term rewards are similar for each validator. This prevents malicious users from accumulating more rewards and hence power, which can compromise the security of the protocol.

\begin{definition}[Reward distribution fairness]
\label{def:rewardFairness}
    Let $\rho_i$ denote the cumulative accumulated reward by validator $V_i$ over a long period of time. A protocol $\Pi$ satisfies \emph{reward distribution fairness} if $\rho_i/\rho_j \approx 1$ for all $i,j \in [n]$. 
\end{definition}

\begin{definition}[Cost Enforceability]
    A protocol is said to have \emph{cost enforceability} if there is a way to enforce the payment of $\cost$ of a transaction $\tx \in \block_i$ \textit{i.e.} in state $\st_i$ the client's resource is deducted by $\cost(\tx,\st_i)$. 
\end{definition}


A secondary aim, but not less relevant, is to guarantee that our modifications do not harm the performance of the blockchain protocol. For this goal we consider the framework intorduced by Amores-Sesar~and~Cachin~\cite{DBLP:conf/esorics/Amores-SesarC24}.
\begin{definition}[Run]
  A \emph{run} is a history with an entry for each round containing the actions, a list of received messages, and a list of sent messages by each process in that round.
 \end{definition} 
This concept is commonly referred to in the distributed algorithms and systems literature as \emph{execution}. However, we use the term \emph{run} to avoid confusion with the execution of transactions.

\begin{definition}[Throughput of runs]
Given a protocol $\Pi$, an adversary \CA, and a run $\CR$, we define the \emph{throughput} of $\Pi$ in the presence of \CA in run $\CR$ as the average gas of executed transactions per round, and we denote it by $\op{throughput}(\Pi,\CA,\CR)$. 
\end{definition}

\begin{definition}[Throughput]
Given a protocol $\Pi$, the \emph{throughput of $\Pi$} is defined as $\op{throughput}(\Pi)$ $:= \inf_{\CA} \E[\op{throughput}(\Pi,\CA,\CR)]$, i.e., 
the infimum over all possible adversaries \CA of the average over the randomness of $\Pi$ of $\op{throughput}(\Pi,\CA,\CR)$ over all possible runs.  
\end{definition}


\begin{definition}[Transaction latency]
Given a protocol $\Pi$, an adversary $\CA$, a run $\CR$, and a transaction $\var{tx}$, we define the \emph{latency} of $\var{tx}$ in the presence of adversary \CA in run $\CR$ as the number of rounds since $\var{tx}$ is $\op{ba-broadcast}$ until the first block containing $\var{tx}$ is executed, and we denote it by $\op{latency}(\Pi,\CA,\CR,\var{tx})$. We define the \emph{latency of $\Pi$} to be the average number of rounds, over the transactions $\var{tx}$ in run $\CR$, since $\var{tx}$ is $\op{ab-broadcast}$ until the first block containing $\var{tx}$ is $\op{bab-delivered}$, and denote it by $\op{latency}(\Pi,\CA,\CR)$.
\end{definition}

\begin{definition}[Latency]
Given a protocol $\Pi$, the \emph{latency} of protocol $\Pi$ is defined as $\displaystyle\op{latency}(\Pi)=\sup_{\CA} \E[\op{latency}(\Pi,\CA,\CR)]$, i.e., the supremum over all possible adversaries \CA of the average over the randomness of the protocol of $\op{latency}(\Pi,\CA,\CR)$ over the possible runs $\CR$. 
\end{definition}

\section{Attacks on the Commons}\label{sec:attacks}
In the following we outline a novel attack we identified in the decoupled model.

In the \emph{gaslighting} attack, the attacker aims to mislead a proposer into accepting transactions that appear to consume significantly more resources (in gas) than they actually do. 
As a result, the proposer includes fewer transactions in the block, believing the capacity constraints have been reached. 
Because the transactions in reality consume far fewer resources than declared, this reduces block capacity utilization, diminishes the proposer’s rewards, and violates reward-distribution fairness. The attacker targets an honest proposer (the \emph{gaslit proposer}) by flooding the mempool with transactions that appear expensive but yield much lower actual rewards. 
As in the decoupled setting, clients may honestly overestimate gas costs it is impossible to distinguish between honest and malicious clients~(see \cref{app:casestudy}). 
This enables attacks that reduce the target proposer's rewards while leaving a larger share of system resources available to other proposers. 
We demonstrate the practical feasibility of this attack in~\cref{app:casestudy} and show that, if Ethereum were to adopt the decoupled model, an attacker could cause aggregate losses exceeding \$300 million.


\subsection{Lower Bounds}
\label{sec:lowerBounds}

This vulnerability leads to fundamental lower bounds on resource utilization and reward fairness in both the coupled and decoupled settings and we derive impossibility results for achieving certain desired properties in the decoupled setting that do not arise in the coupled one.

First, we show that in the coupled setting, we can obtain the best possible resource utilization. 

\begin{restatable}[]{lemma}{coupledCRRes}
\label{lem:coupledCRRes}
    There exists a block-creation algorithm $B$ such that, if the mempool is congested, then the block created by $B$ has $\crRes = 1$ in the coupled setting. 
\end{restatable}
\begin{proof}
    See Appendix~\ref{app:coupledCRRes}
\end{proof}

While the above lemma shows that shared resources are fully utilized in the coupled setting, this does not hold in the decoupled setting. We show that any block creation algorithm can be severely undermined by a gaslighting attack.

\begin{restatable}[]{lemma}{decoupledCRRes}
\label{lem:decoupledCRRes}
    For any block creation algorithm $B$, there exists an adversary that can make $\crRes = \crRes^* = 0$ for a block in the decoupled setting, if the mempool is congested and execution is at least two rounds behind the consensus. 
\end{restatable}
\begin{proof}
    See Appendix~\ref{app:decoupledCRRes}
\end{proof}

In fact, not only can an adversary make $\crRes = 0$, it can do so at no cost under the current model, as we see next. 

\begin{restatable}[]{lemma}{lowCRRewActualCost}
\label{lem:lowCRRewActualCost}
    Even if there is cost enforceability, then the cost to the adversary of an attack in~\cref{lem:decoupledCRRes} in the current cost model is $0$ and $\crRew^* = 0$.
\end{restatable}
\begin{proof}
    See Appendix~\ref{app:lowCRRewActualCost}
\end{proof}
    
In the lemma above, $\crRew^* = 0$ because the cost function is $\cost(\tx) = p(\tx)\cdot \gas(\tx,\st)$, and the adversary can ensure that $\gas(\tx,\st)=0$. This raises a natural question: can a different cost function discourage the attack? Our next result shows that, under current decoupled blockchain designs, \emph{modifying the cost function does not help}, as the adversary can still evade paying the cost.

\begin{restatable}[]{lemma}{noCostEnforceability}
\label{lem:noCostEnforceability}
    There is no cost enforceability in the current decoupled blockchain models if execution is at least two rounds behind the consensus. 
\end{restatable}
\begin{proof}
    See Appendix~\ref{app:noCostEnforceability}
\end{proof}

Finally, we show that as a consequence of~\cref{lem:lowCRRewActualCost,lem:noCostEnforceability}, reward fairness is violated in the decoupled model.

\begin{restatable}[]{theorem}{noRewardFairness}
\label{lem:noRewardFairness}
    In the decoupled model, there is no reward fairness. 
\end{restatable}
\begin{proof}
    See Appendix~\ref{app:noRewardFairness}
\end{proof}

Without cost enforceability, regardless of the function used for $\cost$, protocols in the decoupled model cannot prevent an adversary from launching attacks from accounts with no funds, since the account states are unknown to block proposers. 
Possible solutions include ensuring that account states are fully known at block creation time (i.e., going back to the the coupled setting), or restricting which transactions may be included in a block to mitigate such attacks, which forms the main idea behind our solution to this problem and which we will discuss in the next section.

\section{Partial Coupling}
\label{sec:partialCoupling}

\begin{figure}[ht!]
    \centering \includegraphics[width=\linewidth]{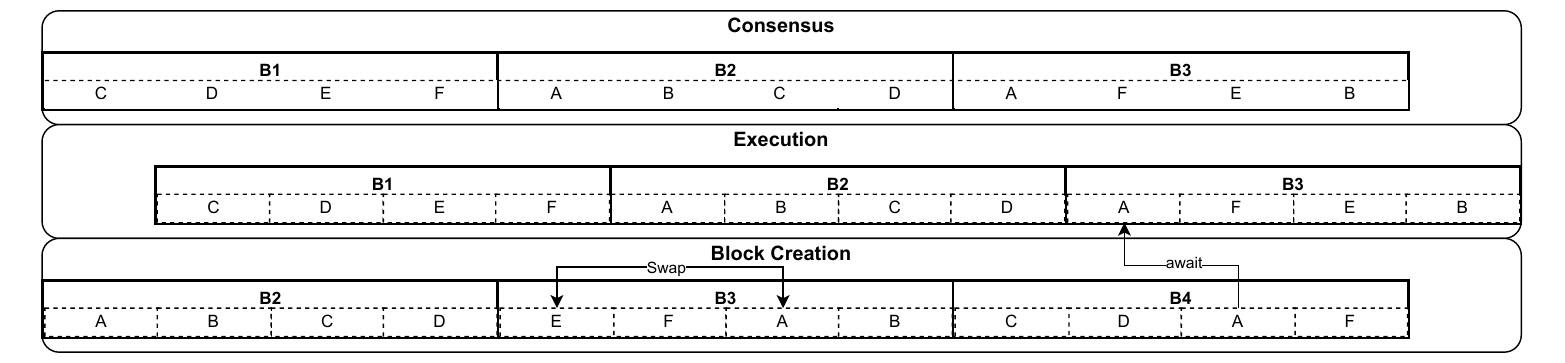}
    \caption{Partial coupling for a set of Blocks $B1, B2, B3, B4$ and sets of tasks accessing resources $A,B,C,D,E,F$. During the creation of Block~B2 (lower left), analysis of a task accessing resource~$C$ can begin once the execution of the previous task accessing~$C$ has completed. Similarly, at the creation of Block~B3, the same applies to tasks accessing resource~$A$. If resource~$A$ is highly contended, such tasks can be placed earlier in the block (e.g., at the beginning of Block~B3), ensuring that a task accessing~$A$ can be analyzed within the block-creation window of Block~B4. Note that execution is slightly shifted relative to consensus, since tentative execution can begin only after the initial block broadcast. }
    \label{fig:partialcoupled}
\end{figure}
We have strongly argued that decoupled systems are at severe risk of attacks and manipulation, which can even lead to long-term stake drift in favor of the adversary.
Here we provide a solution that achieves the best of both worlds: the efficiency of a decoupled system and the incentive mechanism of a coupled system.  And this solution is eminently practical and could be implemented in most decoupled systems today.
Our solution rests on a realistic and reasonable assumption: there are enough independent transactions, with respect to pending executions, in the mempool at any time; this is equivalent to the standard assumption of DAG-Based protocols~\cite{DBLP:conf/podc/KeidarKNS21,narwhal,DBLP:conf/wdag/Amores-SesarGHO25,DBLP:conf/ndss/BabelCDKKKST25} to ensure that blocks are not partially full.

Let $t$ be the last round of execution that finished, when round~$i$ of block creation starts. 
Let $\TaskSet_P(t,i) := \cup_{j \in [t+1,i-1]} \block_j$, where $t < i$, denote the set of pending transactions in blocks whose execution outcomes are not yet finalized. 
The $\block_i$ denotes the set of transactions scheduled for execution after $\TaskSet_P(t,i)$. 
Let $\st_t$ be the execution state known at the start of round~$i$ of block creation. 
The execution of $\TaskSet_P(t,i)$ induces an unknown state transition $\Delta(\TaskSet_P(t,i))$, yielding the system state at execution round $i$ as $\st_{i} = \st_t \oplus \Delta(\TaskSet_P(t,i))$.


Since $\Delta(\TaskSet_P(t,i))$ is unknown at the time of block creation, $\block_i$ is scheduled assuming knowledge of only state~$\st_t$, while execution occurs under the potentially different state~$\st_{i-1}$.

This observation leads to a key insight. If $\block_i$ is restricted so that for each $\tx \in \block_i$, the bits of state that can effect the transition of $\tx$ is disjoint from the state that may be written by $\TaskSet_P(t,i)$, i.e., $\mathsf{Read}(\block_i) \cap \mathsf{Write}(\TaskSet_P(t,i)) = \emptyset$ then the execution of $\TaskSet_P(t,i)$ cannot affect the transition state induced by $\block_i$. Consequently, the validator for round~$i$ of block creation has complete information, and the runtime behavior and outputs of $\block_i$ can be determined deterministically. Once these tasks become pending, their known outputs can be used in subsequent block creation decisions.

As a result, transaction analysis at block creation time becomes fully deterministic, and malicious behavior, such as intentional resource or runtime misestimation, can be detected and punished. The block itself serves as verifiable evidence of misbehavior, enabling penalties.

We formalize this principle as \emph{partial resource coupling}. Rather than coupling execution and consensus at the block level, partial resource coupling requires the scheduler to exclude from $\block_i$ any task whose read set may intersect the write set of $\Delta(\TaskSet_P(t,i))$. Under this constraint, scheduling decisions rely on complete and correct state information despite decoupled execution. Figure~\ref{fig:partialcoupled} outlines how this could look in practice.

Finally, after analyzing $\block_i$, the validator may reorder tasks to prioritize access to frequently read state objects. This increases their availability to subsequent validators and reduces the inclusion latency of tasks accessing popular resources. Such reordering preserves correctness as long as the affected transactions access independent bits of the execution state, as the two transactions that access different resources produce the same output regardless of the order of execution.

\begin{restatable}[]{theorem}{partialcoupled}
\label{thm:partialcoupledsecurity}
The partially coupled setting guarantees security against gaslighting.

\end{restatable}
\begin{proof}
See Appendix~\ref{app:partialcoupledsecurity}
\end{proof}

Finally, as we prove in the following, in the worst case, the partially coupled model only has a minor latency overhead compared to the coupled model, and, as long as there are some transactions independent of pending blocks to include, partial coupling always increases the system throughput and improves latency compared to the coupled model up to the same throughput as the decoupled model as we can overlap consensus, execution and block creation.

\begin{restatable}[]{theorem}{partialcoupledlatency}
\label{thm:partialcoupledlatency}
The difference in latency between the partially coupled and coupled settings is at most $\frac{\delta^c}{2}$. Furthermore, if in every round there exists a set $Tr$ in the mempool, with $N$ transactions, which are independent of pending blocks and $\gas(Tr,\st_t) \geq \maxgas$, then the partially coupled setting has lower latency than the coupled setting.
\end{restatable}

\begin{proof}
See Appendix~\ref{app:partialcoupledlatency}
\end{proof}

\begin{restatable}[]{theorem}{partialcoupledthroughput}
\label{thm:partialcoupledthroughput}
The partially coupled setting has the same throughput as the decoupled setting if in every round there exists a set $Tr$ in the mempool, with $N$ transactions, which are independent of pending blocks and $\gas(Tr,\st_t) \geq \maxgas$. Further, the partially coupled setting has $\crRes = 1$.
\end{restatable}

\begin{proof}
See Appendix~\ref{app:partialcoupledthroughput}.
\end{proof}

Both results above assume that the mempool always contains a sufficient number of transactions. To obtain an estimate, suppose that each transaction has unit gas cost, so that at most $\maxgas$ transactions fit in a block. Assume further that each transaction accesses a constant number of states, and that execution lags at most a constant factor $C$ behind consensus. Under these assumptions, our condition becomes the following: there exist $\maxgas$ transactions that do not conflict with the $C\maxgas$ transactions pending execution, that is, a constant number of non-conflicting transactions.

By contrast, DAG-based protocols~\cite{DBLP:conf/podc/KeidarKNS21,narwhal,DBLP:conf/wdag/Amores-SesarGHO25,DBLP:conf/ndss/BabelCDKKKST25} assume that, at any point in time, the mempool contains $n^2$ non-conflicting transactions, where $n$ is the number of validators. We therefore stress that our assumption is significantly less restrictive.

To conclude, note that the partially coupled setting trivially achieves higher throughput than the coupled setting. This follows from the fact that partial coupling attains the same throughput as the decoupled model, which in turn outperforms the coupled setting in terms of throughput.
\section{Discussion, Limitations and Future Work}
\label{sec:discussion}
We now discuss the practicalities and limitations of gaslighting and partial coupling (more details in~\Cref{discussion:gaslighting,discussion:partial} respectively). We summarize our discussion as follows:

\begin{itemize}
    \item Gaslighting: We first observe that the damage done by gaslighting is augmented in the decoupled setting with parallel execution of transactions. This would, in the worst case, not only be a gaslighting attack, but also negate any throughput benefit of executing transactions in parallel. Furthermore, we discuss potential probabilistic solutions as well as solutions involving gas escrow to mitigate the impact of gaslighting.
    \item Partial coupling: We discuss limitations of partial coupling (e.g., in leaderless protocols), as well as observe that the throughput of partial coupling increases if the block proposer remains the same across consecutive blocks. This, however, trades censorship resilience for throughput.
    Furthermore, we discuss an extension of partial coupling that accounts for resources that are frequently used in order for future block proposers to more accurately estimate gas consumption and resource conflicts.
    Finally, we discuss mechanism design details against attacks in the partial coupled setting that access a wide range of resources. In particular, we note that the deterministic guarantees of partial coupling with respect to resource access and execution order makes it easier to design resource pricing mechanisms.
\end{itemize}

\subsection{Gaslighting}\label{discussion:gaslighting}
\smallskip{\bf Gaslighting and Parallel Execution.}
Many modern blockchains execute transactions in parallel rather than sequentially, which can reduce the effective execution time (or CPU time) by a factor of $c$, assuming at least $c$ cores are available on all validator machines. 
However, this assumes a perfectly parallelizable workload. In practice, as studies have shown~\cite{anthemius,chiron}, blockchain workloads exhibit high contention, with numerous transaction interdependencies that hinder parallel execution. 
To address this, blockchains such as Sui and Solana~\cite{sui,solana} require users to explicitly specify transaction dependencies to enable efficient scheduling.

As a result, at block creation time the proposer must account not only for the sequential gas limit but also for transaction dependencies to achieve high block capacity utilization. 
Otherwise, a long sequential chain of transactions can negate the performance benefits of parallel execution~\cite{anthemius}. 

This, in turn, expands the attack surface for gaslighting: a client can not only overestimate its resource usage in gas but also deliberately inflate dependencies, impairing parallelism and thereby reducing block utilization.

\smallskip{\bf Probabilistic Attacks.}
The success of gaslighting depends on network and system conditions: if a client’s transaction is included in an unexpected block, the attack may fail. Although the partially synchronous model~\cite{dls} often assumes some adversarial network control, real deployments are more nuanced. While deterministically preventing these attacks in the decoupled model requires a trade-off between throughput and resilience to gaslighting, introducing randomness could yield probabilistic defenses that render the attacks impractical. We leave the design of such randomness-based solutions to future work.

\smallskip{\bf Client Staking.}
One limitation of the gaslighting attack model is that clients may transfer their funds before being fully charged for the submitted transaction. This enables malicious clients to reserve execution capacity and fill blocks without ultimately paying for the resources they occupy.

A natural mitigation is to require clients to escrow gas upfront, reserving a multiple of the declared maximum gas for each submitted transaction. If multiple block proposals may proceed in parallel, the required reservation must scale with the maximum number of concurrent proposals. Such a mechanism prevents clients from evading payment by depleting their balance prior to final settlement.

However, this approach does not eliminate the broader class of gaslighting attacks. As discussed, the system cannot reliably distinguish between malicious overestimation of resource usage and benign misestimation due to limited state knowledge. Consequently, preventing strategic misestimation would require charging clients for the full reserved resource allocation, independent of actual consumption.

While reservation-based pricing is common in cloud environments, blockchains represent shared public resources with an inherent objective of high utilization rather than revenue maximization for validators. Charging for full reservation therefore negatively affects client welfare, as resources remain underutilized while clients willing to pay for resource allocation experience increased waiting times. This trade-off is closely related to what is commonly referred to as the price of anarchy~\cite{selfish}.

\subsection{Partial coupling}\label{discussion:partial}
\smallskip{\bf Leaderless Protocols.}
While partial coupling mitigates the new attack vectors we identified in this paper, partial coupling can only be applied to leader-based consensus. In leaderless protocols, where nodes may propose blocks in parallel, enforcing synchronization between the proposers would effectively collapse the design into a leader-based protocol. Techniques such as address-space forking, as in MirBFT~\cite{mirbft}, can mitigate some of these issues, but it does not address the broader class of gaslighting attacks. We therefore defer a detailed analysis of leaderless protocols to future work.

\smallskip{\bf Stable Leaders.}
In the partially coupled setting, in practice, the primary source of delay during block creation arises from executing transactions in order to accurately determine gas consumption and resource accesses. If the block proposer remains the same process across multiple consecutive rounds, it can reuse knowledge of the output state of pending blocks it created itself. Consequently, it may include transactions accessing the same resources in subsequent blocks without being subject to gaslighting.

While this approach has demonstrated numerous advantages in terms of throughput~\cite{kauri}, it introduces several important trade-offs.
For once, if a leader is stable over multiple rounds, clients might be censored for a prolonged amount of time. Furthermore, the known leader becomes susceptible to DDoS attacks, that can impact the system performance. As such, this results in a trade-off between censorship resilience and throughput.

\smallskip{\bf Extended Partial Coupled Model.}
The basic partially coupled model introduced in this paper enables consensus and execution to proceed in parallel, thereby increasing execution throughput and reducing consensus latency without exposing the system to the attacks identified for fully decoupled designs.

This benefit comes with the caveat that transactions accessing overlapping state may need to be deferred to subsequent blocks, particularly in comparison to the decoupled model. One possible mitigation is to include, within each block, the output state corresponding to frequently used resources. This allows the next block proposer to use these outputs to accurately estimate gas consumption and resource conflicts for transactions in the mempool.

If a malicious proposer includes an incorrect output state, this deviation can be detected deterministically during execution, enabling the protocol to slash the proposer for Byzantine behavior. The primary overhead of this approach is the additional data that must be included in each block, which increases the bandwidth cost of consensus.

A formal analysis of this approach, as well as other potential extensions of the partially coupled model, is left to future work.

\smallskip{\bf Game Theoretic Improvements.}
In the current design of the partially coupled model, an adversary could attempt to degrade throughput by submitting transactions that access a wide range of system resources, which reduces the availability of independent transactions for subsequent blocks. This makes it harder for the next block proposer to construct a block, as the majority of the transactions in the mempool might conflict with this transaction.

This attack vector is not unique to partial coupling. Any parallel execution model is vulnerable to transactions with broad resource accesses, which can force sequential execution of block and limit the achievable speed-up.

However, partial coupling provides determinism guarantees with respect to resource access and execution order. These guarantees enable the design of pricing mechanisms in which transactions that access resources recently used in pending or prior blocks incur higher fees. Such resource-aware pricing can render the attack economically unattractive while preserving throughput under honest workloads.

\bibliographystyle{ACM-Reference-Format}
\bibliography{refs}

\appendix

\section{Omitted Proofs}

\subsection{Proof of~\Cref{lem:coupledCRRes}}\label{app:coupledCRRes}
\coupledCRRes* 

\begin{proof}
    Since the mempool is congested, there exists $S \subset M_{i-1}$ such that $|S| \leq N$ and $\gas(S) = \maxgas$. The block creation algorithm simply finds this set, assuming access to a knapsack oracle, and puts $\block_i = S$. This ensures that $\gas(\block_i) = \maxgas$. Thus $\crRes = 1$.  
\end{proof}

\subsection{Proof of~\Cref{lem:decoupledCRRes}}\label{app:decoupledCRRes}
\decoupledCRRes* 

\begin{proof}
    Since block creation is deterministic, the adversary $\adv$ adds two disjoint set of transactions, $S_\adv$ and $S^\prime_\adv$, to the mempool $M_{i-1}$ such that:
    \begin{enumerate}
        \item $B(M_{i-1}) = \block_i \subseteq S_\adv$, where $\est(\tx) = \maxgas/N$ and $|S_\adv| \geq N$ but $\gas(\tx, E(\st_{i-1}, \block_i)) = 0$ for all $\tx \in S_\adv$. 
        \item For any $\tx \in S^\prime_\adv$, $\est(\tx) = \gas(\tx, \st_i) = 1$. 
    \end{enumerate}
    Thus adversary succeeds in getting $\crRes = 0$ for block $\block_i$. At the same time, if the state $\st_i$ was known then best block $\block^*_i$ for resource utilization would have $\block^*_i \subset S^\prime_\adv$. Thus $\crRes* = 0$ for $\block_i$. 
\end{proof}


\subsection{Proof of~\Cref{lem:lowCRRewActualCost}}\label{app:lowCRRewActualCost}
\lowCRRewActualCost* 

\begin{proof}
    Any transaction, $\tx$, in the proof of~\cref{lem:decoupledCRRes} has $\gas(\tx, \st_i) = 0$ where $\st_i = E(\st_{i-1}, \block_i)$. Thus in current cost model the $\cost(\tx, \st_i) = p(\tx)\cdot\gas(\tx,\st_i) = 0$.  
\end{proof}

\subsection{Proof of~\Cref{lem:noCostEnforceability}}\label{app:noCostEnforceability}
\noCostEnforceability* 

\begin{proof}
    At any round $i$, the resource of a client is defined by the execution state $\st_i$ after all transactions in blocks upto $\block_i$ are executed. 
    If the execution is two rounds behind the block creation, then an adversary sends a transaction $\tx$, which transfers all its resources, in the round $i-2$ to the mempool $M_{i-2}$ with a price $p(\tx)$ high enough such that $B_{i-2}(M_{i-2})$ puts $\tx \in \block_{i-2}$. 
    Then the adversary would launch the attack from~\cref{lem:decoupledCRRes} in round $i$. 
    Thus, it cannot be made to pay for its transactions.   
\end{proof}

\subsection{Proof of~\Cref{lem:noRewardFairness}}\label{app:noRewardFairness}
\noRewardFairness* 

\begin{proof}
    A rational validator would launch the attack described in the~\cref{lem:decoupledCRRes} in every round, and we see from~\cref{lem:lowCRRewActualCost,lem:noCostEnforceability} that the validator submitting such a transaction would pay $0$ for the attack. For the round where the validator is the proposer, it would not include its own submitted transactions to its block. Thus, for every other validator reward is $0$, while the reward for the rational validator is greater than $0$. This violates reward fairness.  
\end{proof}

\subsection{Proof~of~\Cref{thm:partialcoupledsecurity}}\label{app:partialcoupledsecurity}
\partialcoupled*
\begin{proof}
Consider an honest validator $S$ for round $i$. By definition, $S$ includes in $\block_i$ only those transactions for which it has complete knowledge of all input states that will be observed at execution time $t+\delta$. Knowing the input state, for each such task, $S$ can deterministically derive its set of write addresses (\textit{e.g.} by executing the transaction) and include these addresses in $\block_i$.

Assume first that all previously constructed blocks were produced by honest validators. Let $S^\prime$ be any subsequent validator. Since each prior block $\block_k$ contains the complete set of write addresses corresponding to $\block_k$, $S^\prime$ can deterministically determine whether its local task set $\block_j$ conflicts with any pending $\block_k$. Thus, conflict detection is sound and complete with respect to all previously scheduled but not yet executed tasks.

Now consider the case where a block is constructed by a malicious validator $S_e$ that declares an incorrect state access set, i.e., it omits required state accesses or declares unnecessary ones. Such inconsistencies are detected at execution time, as the actual state accesses performed by the tasks will not match the declared access set. Upon detection, the corresponding transaction is aborted, and $S_e$ is excluded from further scheduling.

Since an aborted transaction does not create state updates, it does not affect the state assumptions under which any pending block $\block_k$ constructed by a correct scheduler was formed. Consequently, all blocks constructed by correct schedulers are guaranteed to have been created with respect to the settled state.

Therefore, any correct scheduler has access to the settled state for each individual transaction it includes in a block.

Since every correct scheduler has access to the settled state for each individual transaction, the scheduler can correctly compute the amount of gas needed to execute this transaction using the gas function. 
This allows the scheduler to check for inconsistencies between the actual gas needed per transaction and the estimated cost.
Assuming a congested mempool, the scheduler can select transactions that minimize these inconsistencies, enabling the scheduler to be secure against gaslighting.

\end{proof}

\subsection{Proof of~\Cref{thm:partialcoupledlatency}}\label{app:partialcoupledlatency}
\partialcoupledlatency*

\begin{proof}
Given a blockchain protocol $\Pi$ and a transaction $\tx$. 

In the coupled setting, scheduling, consensus, and execution proceed sequentially. Thus, the latency of transaction $\tx$ is
$\op{latency}(\tx)=\delta^w+\delta^b+\delta^c+\delta^e$,
where $\delta^{w}(\tx)$ is the expected time from when the transaction joins the mempool until it is added to a block.

In the partially coupled setting, scheduling proceeds concurrently with consensus and execution. Therefore, the latency of $\tx$ satisfies
$\op{latency}(\tx)<\delta^{w'}(\tx)+\delta^b+\delta^c+\delta^e$,
where $\delta^{w'}(\tx)$ may differ from $\delta^{w}(\tx)$.

The difference $\delta^{w'}(\tx)-\delta^w(\tx)$ is positive when $\tx$ accesses the same states as transactions agreed upon by consensus but still not executed.  

If there exists a set $Tr$ at every round with $N$ transactions in the mempool which are independent of pending blocks and $\gas(Tr,\st_t) \geq \maxgas$ , then every block is at full capacity. This implies $\crRes = \crRes^* = 1$. Further, for every additional block in which $\tx$ is not included, there exists a transaction $\tx'$ whose waiting time $\delta^{w'}(\tx')$ is reduced by the same amount.
Thus, when taking expectations over all transactions, as part of the randomness of all runs, we conclude that the latency of $\Pi$ in the partially coupled setting is lower than in the coupled setting.

If there exists no set $Tr$ at some round with $N$ transactions in the mempool which are independent of pending blocks and $\gas(Tr,\st_t) \geq \maxgas$, some block may be only partially filled, since certain transactions must wait for the execution of conflicting transactions to complete. This is a common issue in the coupled setting, with the key difference that, during the waiting period, the partially coupled setting continues running consensus instances. As a result, on average, such transactions wait $\frac{\delta^c}{2}$ longer than in the coupled setting. Considering the worst-case scenario in which all transactions conflict, we conclude that the difference between the latency of the partially coupled and coupled settings is at most $\frac{\delta^c}{2}$.
\end{proof}

\subsection{Proof of~\Cref{thm:partialcoupledthroughput}}\label{app:partialcoupledthroughput}
\partialcoupledthroughput*

\begin{proof}
Note that in both the decoupled and partially coupled settings, block creation, consensus, and execution are performed in parallel and require the same amount of time. Thus, the only parameter that can affect throughput is the gas of the transactions included in each block. However, as long as in every round there exists a set $Tr$ in the mempool, with $N$ transactions, which are independent of pending blocks and $\gas(Tr,\st_t) \geq \maxgas$, blocks in the decoupled and partially coupled settings are both filled with the same amount of gas. Hence, under this assumption, the partially coupled setting achieves the same throughput as the decoupled setting. It follows that $\crRes = 1$. 
\end{proof}

\section{Case Study on gaslighting in Ethereum and Sui}
\label{app:casestudy}

Figure~\ref{fig:pseudocode_comparison} illustrates the attack outlined in Section~\ref{sec:attacks} using two example smart contracts. Both provide a cheap or expensive execution path depending on a variable $x$, yet are indistinguishable under static analysis. This shows that gaslighting cannot be mitigated through static analysis alone: honest interactions can produce misleading gas estimates. As long as attackers replicate real-world contracts and substitute analogous cheap and expensive operations, the attack remains undetectable by static analysis.

\begin{figure}[t!]
    \centering
    \begin{minipage}{0.45\textwidth}
    
        \begin{algorithm}[H]
        \scriptsize
        \caption*{NFT Minting}
        \begin{algorithmic}[1]
        \State $x = 0$
        \Function{mint\_nft}{}
            \If{$x > 100$}
                \State \textsc{Refund-User} \Comment{Cheap op}
                \State \Return
            \Else
                \State $x \gets x + 1$
                \State \textsc{Mint NFT} \Comment{Expensive op}
            \EndIf
        \EndFunction
        \Function{reopen-minting}{$start_{idx}$}
        \State{$x=start_{idx}$}
        \EndFunction
        \end{algorithmic}
        \end{algorithm}
    \end{minipage}
    \hfill
    \begin{minipage}{0.45\textwidth}
    
        \begin{algorithm}[H]
        \scriptsize
        \caption*{Attack Contract}
        \begin{algorithmic}[1]
        \State $x = 0$
        \Function{gaslight}{}
            \If{$x > 100$}
                \State \textsc{Cheap-Op}
                \State \Return
            \Else
                \State $x \gets x + 1$
                \State \textsc{Expensive-Op}
            \EndIf
        \EndFunction

        \Function{set-x}{$input_x$}
        \State{$x=input_x$}
        \EndFunction
        \end{algorithmic}
        \end{algorithm}
    \end{minipage}

    \caption{Comparison of NFT Minting vs Attack Smart Contract.}
    \label{fig:pseudocode_comparison}
\end{figure}

Furthermore, to underline the feasibility of these attacks, we implemented\footnote{\url{https://anonymous.4open.science/r/commons-analysis-46BB}} an example Solidity smart contract based on the pseudocode in Figure~\ref{fig:pseudocode_comparison} and measured the difference in gas costs between the two execution paths depending on the value of $x$. 
We deployed the contract in a local testnet and observed a gas consumption of $23{,}373$ for the early exit and $16{,}722{,}766$ (very close to the maximum per-transaction gas limit on Ethereum) for the long path. 
This corresponds to a factor of $715$, demonstrating that such attacks are indeed feasible in practical smart contract environments. 

In addition, we analyzed the discrepancy between the user-specified maximum gas and the actual gas consumed by transactions in practice, as this gap significantly affects the efficiency of the decoupled model. 
Specifically, we analyzed all Ethereum transactions invoking smart contracts between June and August 2025, totaling $90{,}223{,}198$ transactions. We found that, on average, declared gas exceeded the actual consumption by $63.9\%$.  
Given a fixed per-block gas limit and using max-gas to bound the maximum block complexity, this discrepancy reduces the speed-up of the decoupled approach from $5$ to $1.8$, which is already below the speed-up achievable with the partially coupled approach.  
Moreover, the 10th percentile of transactions in terms of overestimation exceeded the actual gas usage by more than $378\%$, and the 25th percentile still overestimated the gas usage by over $114\%$.

We further analyzed one million Sui checkpoints (from checkpoint  189,000,000 to 190,000,000) and measured the average gas usage per transaction~\footnote{\url{https://anonymous.4open.science/r/sui-A67E/crates/sui-data-ingestion-core/src/reader.rs\#L404}}
. On average, gas estimates on Sui exceeded actual consumption by $400\%$, with the lower 10th percentile still overestimating by $151\%$, much higher than the average overestimation observed on Ethereum. At the upper end, the 25th to 10th percentiles overestimated their usage by between $6{,}273\%$ and $20{,}203\%$.
These results highlight that, compared to Ethereum, Sui’s decoupled model exhibits substantially larger discrepancies between estimated and actual gas usage, making it significantly harder to mitigate attacks, as most transactions already vastly overestimate their gas requirements.

Additionally, we analyze the potential financial impact of a gaslighting attack. 
The attacker's success and the resulting cost depend primarily on the factor by which an adversary can maliciously overestimate gas.  
On Ethereum, we implemented an example that produced a factor of $\times715$. On Sui, the documented maximum gas budget is $50$ billion, while a simple smart-contract call may cost only $1{,}075{,}000$ gas; this yields a worst-case overestimation factor of approximately $\approx 46{,}512$~\cite{suigas}.

The financial impact of such an attack is determined mainly by the size of the rewards that accrue from transaction fees, which we refer to as \emph{execution-layer rewards}.  
At the time of writing, execution-layer rewards on Ethereum correspond to an APR of approximately $0.22\%$~\cite{RatedNetworkExplorer2025}.  
Given the total USD value of ETH currently staked at \$160.1 billion, the annual execution-layer rewards are approximately \$352 million. We use Ethereum as a reference point for execution layer rewards, as these are well documented and publicly available~\cite{RatedNetworkExplorer2025}.

Next, we estimate the potential impact of gaslighting attacks.  
In the worst case, considering the total execution-layer rewards and a maximum gas overestimation factor of $715$, a consistently successful gaslight attacker could capture a $\frac{714}{715}$ fraction of these rewards.  

For a $10\%$ adversary, this corresponds to approximately \$316 million under normal operation, and \$632 million under network congestion. During congestion, priority fees rise; it may take up to $10$ gwei for a transaction to be included by a proposer. We estimate that a $50$ gwei priority fee would ensure inclusion even under high congestion.  

Given that a small smart contract can be executed with $50{,}000$ gas, this places the cost of attacking each block for a year at roughly $6{,}500$ ETH (\$29 million) under worst-case conditions.  
Meanwhile, average execution-layer rewards per block are $0.022$ ETH, with approximately $2.6$ million blocks per year. Thus, the attacker could gain close to \$40 million in normal circumstances and up to \$110 million during periods of high congestion.

Although this represents the worst-case impact, block proposers and clients would likely notice anomalous behavior.  
We also estimate the effect of a ``silent attack,'' where an attacker fills only $10\%$ of the block space with low-value or garbage transactions.  
Compared to the worst case, this would result in approximately \$31.6 in damages, while the cost to the attacker would be roughly \$2.9 million.









\end{document}
\endinput